\documentclass[11pt,a4paper]{article}

\usepackage{amsmath}
\usepackage{amsthm}
\usepackage{amssymb}
\usepackage{algorithmic}
\usepackage{algorithm,enumitem}
\usepackage{amsopn}
\usepackage{mathtools}
\usepackage{authblk}
\usepackage{geometry}   
\newtheorem{assumption}{Assumption}
\newtheorem{observation}{Observation}
\newtheorem{lemma}{Lemma}
\newtheorem{remark}{Remark}
\newtheorem{theorem}{Theorem}
\newtheorem{corollary}{Corollary}
\newtheorem{definition}{Definition}

\newcommand*{\transp}{\mathsf{T}}
\newcommand*{\Z}{\mathbb{Z}}
\newcommand*{\Q}{\mathbb{Q}}

\newcommand{\tabincell}[2]{\begin{tabular}{@{}#1@{}}#2\end{tabular}}

\begin{document}
\title{Approximation Algorithms for Hard Capacitated $k$-facility Location Problems}
\author[1,2]{Karen Aardal}
\author[1]{Pieter van den Berg}
\author[1]{Dion Gijswijt}
\author[1]{Shanfei Li \thanks{Corresponding authors: shanfei.li@tudelft.nl}}
\affil[1]{Delft Institute of Applied Mathematics, Delft University of Technology, The Netherlands}
\affil[2]{Centrum Wiskunde en Informatica, The Netherlands}
\date{}

\maketitle
\begin{abstract}
We study the capacitated $k$-facility location problem, in which we are given a set of clients with demands, a set of facilities with capacities and a constant number $k$. It costs $f_i$ to open facility $i$, and $c_{ij}$ for facility $i$ to serve one unit of demand from client $j$. The objective is to open at most $k$ facilities serving all the demands and satisfying the capacity constraints while minimizing the sum of service and opening costs.

In this paper, we give the first fully polynomial time approximation scheme (FPTAS) for the single-sink (single-client) capacitated $k$-facility location problem. Then, we show that the capacitated $k$-facility location problem with uniform capacities is solvable in polynomial time if the number of clients is fixed by reducing it to a collection of transportation problems. Third, we analyze the structure of extreme point solutions, and examine the efficiency of this structure in designing approximation algorithms for capacitated $k$-facility location problems.
Finally, we extend our results to obtain an improved approximation algorithm for the capacitated facility location problem with uniform opening cost.
\end{abstract}

\section{Introduction}
In the capacitated $k$-facility location problem (CKFL), we are given a set $D$ of clients and a set $F$ of potential facilities (locations where we can potentially open a facility) in a metric space. Each facility $i\in F$ has a capacity $s_i$. Each client $j$ has a demand $d_j$ that must be served. Establishing facility $i$ incurs an opening cost $f_i$. Shipping $x_{ij}$ units from facility $i$ to client $j$ incurs service costs $c_{ij} x_{ij}$, where $c_{ij}$ is proportional to the distance between $i$ and $j$. The goal is to serve all the clients by using at most $k$ facilities and satisfying the capacity constraints such that the total cost is minimized. In this paper, we consider the \emph{hard} capacities, that is, we allow at most one facility to be opened at any location. (Note that in the \emph{soft} capacities case multiple facilities can be opened in a single location \cite{MahdianYZ}.)

CKFL can be formulated as the following mixed integer program (MIP), where variable $x_{ij}$ indicates the amount of the demand of client $j$ that is served by facility $i$, and $y_{i}$ indicates whether facility $i$ is open.

\begin{alignat}{2}
\min\quad&\sum_{i\in F}\sum_{j\in D}c_{ij} x_{ij}+\sum_{i\in F}f_i y_i \label{eq:MIP}\\
\text{subject to:}\quad&\sum_{i\in F}{x_{ij}} =  d_j,&\forall j\in D,\label{eq:demand constraint}\\
&\sum_{j\in D}{x_{ij}} \leq  s_i y_i,&\forall i\in F,\label{eq:capacity constraint}\\
&\sum_{i\in F}{y_i} \leq  k,\label{eq:cardinality constraint}\\
&x_{ij}\geq  0,&\forall i\in F, \forall j\in D,\label{eq:nonegative constraint}\\
&y_{i} \in \{0,1\},&\forall i\in F.\label{eq:binary constraint}
\end{alignat}
If we replace constraints (\ref{eq:binary constraint}) by
\begin{alignat}{2}
0 \leq y_{i}\leq 1, i\in F, \label{binary relaxation}
\end{alignat}
we obtain the LP-relaxation of CKFL. Without loss of generality we suppose that $s_i$, $d_j$ for each $i\in F, j\in D$ are all integral.

CKFL is related to the capacitated $k$-median problem (CKM), in which inputs and goal are the same as CKFL except that there is no opening cost for facilities.
A constant factor approximation algorithm is still unknown for CKM, let alone CKFL.
All the previous attempts with constant approximation ratios
for these problems violate the capacity constraint, or cardinality constraint that at most $k$ facilities are allowed to be used.
We call these approximation algorithms \textit{pesudo-approximation algorithms}.
Recently, Byrka et al. \cite{ByrkaFRS} gave a constant factor approximation algorithm for CKM with uniform capacities
while violating the capacities with a factor $2+\epsilon$, where $\epsilon>0$ can be arbitrarily small.
Although most researchers believe that relaxing the cardinality constraint makes the problem simpler
than relaxing the capacity constraint with respect to designing pesudo-approximation algorithms,
the best known violation ratio for cardinality constraint is still $5+\epsilon$ to get a
constant factor approximation algorithm \cite{KorupoluPR} for CKM with uniform capacities.
It seems that to obtain a better constant factor approximation algorithm with violating the cardinality constraint
has not received much attention yet.

In this paper, we give an improved approximation algorithm for CKFL (with arbitrary capacities) with uniform opening cost by using at most $2k$ facilities.
To show the potential power of this algorithm, we improve the approximation ratio for the capacitated facility location problem
with uniform opening cost \cite{LeviSS}, by combining this algorithm with a pesudo-approximation algorithm for the $k$-median problem derived from
a bifactor approximation algorithm for the uncapacitated facility location problem \cite{CharikarG}.
That is, pesudo-approximation algorithms for capacitated $k$-facility location problems may be extended to get
approximation algorithms for well-studied capacitated facility location problems.
We believe that this technique has the potential to further improve approximation ratios for capacitated facility location problems.

Additionally, in Section \ref{sec:SCKF} we give the first fully polynomial time approximation scheme (FPTAS) for
the single-sink (hard) capacitated $k$-facility location problem.
In Section \ref{sec:CKFLU}, we give a {polynomial time} algorithm for the uniform capacitated $k$-facility location problem with a fixed number of clients.

\subsection{Related Work}
The $k$-facility location problem has already been studied since the early 90s \cite{CornuejolsNW,HsuLT}. It is a common generalization of the $k$-median problem (in which at most $k$ facilities are allowed to be opened, and there is no opening costs) and the uncapacitated facility location problem, which are classical problems in computer science and operations research and have a wide variety of applications in clustering, data mining, logistics \cite{BradleyFM,JainD,KuehnH}, even for the single-sink (single client) case \cite{HererRH}.

For the uncapacitated $k$-facility location problem (UKFL), Charikar et al. \cite{CharikarGTS} gave the first constant factor approximation algorithm with performance guarantee 9.8, by modifying their $6\tfrac{2}{3}$-approximation algorithm for the uncapacitated $k$-median problem.
Later, the approximation ratio was improved by Jain and Vazirani \cite{JainV}, who made use of a primal-dual scheme and Lagrangian relaxation techniques to obtain a $6$-approximation algorithm.
Jain et al. \cite{JainMMSV,JainMS} further improved the ratio to $4$ by using a greedy approach and the so-called Lagrangian Multiplier Preserving property of the algorithms. The best known approximation algorithm for this problem, due to Zhang \cite{Zhang}, achieves a factor of $2+\sqrt{3}+\epsilon$ using a local search technique.
The $k$-median problem, as a special case of UKFL, was studied extensively \cite{ArcherRS,AryaGKMMP,ByrkaPRST, CharikarG,CharikarGTS,JainMS,JainV,LiS} and the best known approximation algorithm was recently given by Byrka et al. \cite{ByrkaPRST} with approximation ratio $2.611+\epsilon$ by improving the algorithm of Li and Svensson \cite{LiS}.
In addition, Edwards \cite{Edwards} gave a $7.814$-approximation algorithm for the multi-level uncapacitated $k$-facility location problem by extending the $6\frac{2}{3}$-approximation algorithm by Charikar et al. \cite{CharikarGTS} for the uncapacitated $k$-median problem.

Unfortunately, the capacitated $k$-facility location problem is much less understood although the presence of capacity constraints
is natural in practice. The difficulty of the problem lies in the fact that two kinds of hard constraints appear together: the cardinality constraint, and the capacity constraints. This seems to result in hardness of the methods such as LP-rounding, primal-dual method used to solve the $k$-median problem, and even local search algorithms used to solve the capacitated facility location problem and the $k$-median problem.

The capacitated $k$-facility location problem is related to the capacitated facility location problem (CFL), whose inputs and goal are the same as for CKFL but without the cardinality constraint. Most known approximation algorithms for CFL are based on local search technique since the natural linear programming relaxation has an unbounded integrality gap for the general case \cite{PalTW}. For nonuniform capacities, P{\'a}l, Tardos, and Wexler \cite{PalTW} proposed the first constant factor approximation algorithm with a factor of $8.53$. Later, Mahdian and P{\'a}l \cite{MahdianP} improved this factor to $7.88$. Zhang, Chen, and Ye \cite{ZhangCY} reduced this factor to $(3+2 \sqrt{2}+\varepsilon)$ by introducing a multi-exchange operation.
The currently best known approximation algorithm, due to Bansal, Garg, and Gupta \cite{BansalGG}, achieves the approximation ratio $5$. As it was expected that the problem is easier for uniform capacities, Korupolu, Plaxton, and Rajaraman (KPR) \cite{KorupoluPR} gave the first constant factor approximation algorithm with a factor of $8$. Later, this factor was improved to $5.83$ by Chudak and Williamson \cite{ChudakW}. The currently best approximation algorithm due to Aggarwal et al. \cite{AggarwalLBGGGJ} has performance guarantee of 3.

Additionally, Levi, Shmoys, and Swamy \cite{LeviSS} showed that the linear programming relaxation has a bounded integrality gap for CFL with uniform opening costs, and gave a $5$-approximation algorithm for this case by an LP-rounding technique.

The capacitated $k$-median problem (CKM), which is a special case of CKFL, is already difficult to handle.
The natural linear programming relaxation has an unbounded integrality gap (see Remark \ref{unbounded gap of SCKFL}).
We have to blow up the capacity or increase the number of opening facilities by a factor of at least 2 if we use the cost of the LP solution as a lower bound to obtain an integral solution \cite{CharikarGTS}.

For the hard uniform capacity case, Charikar et al. \cite{CharikarGTS} gave a constant factor approximation algorithm while violating the capacities within a constant factor $3$ by LP-rounding. Recently, Byrka et al. \cite{ByrkaFRS} improved this violation ratio to $2+\epsilon$ by designing a $(32 l^2+28l+7)$-approximation algorithm increasing the capacity by a factor of $2+\frac{3}{l-1}$, $l\in \{2,3,4,\cdots\}$. Based on a local search technique, Korupolu et al. \cite{KorupoluPR} proposed a $(1+{5}/{\epsilon})$-approximation algorithm by using at most $(5+\epsilon)k$ facilities, and a $(1+{\epsilon})$-approximation algorithm by using at most $(5+5/\epsilon)k$ facilities.

For soft non-uniform capacities, based on primal-dual and Lagrangian relaxation methods, Chuzhoy and Rabani \cite{ChuzhoyR} presented a $40$-approximation algorithm by violating the capacities within a constant factor of $50$.
Bartal et al. \cite{BartalCR} proposed a $19.3(1+\delta)/\delta^2$-approximation algorithm ($\delta>0$) by using at most $(1+\delta)k$ facilities.

To the best of our knowledge, for hard non-uniform capacities, a constant factor approximation algorithm is still unknown
if we allow for a violation of the two kinds of hard constraints: the cardinality constraint and capacity constraints. Without violating any constraint, a constant factor approximation algorithm remains unknown even for the single-sink capacitated $k$-median problem in which $|D|=1$, let alone the capacitated $k$-facility location problem.

\subsection{Our Contributions and Techniques}

(\romannumeral1) The single-sink facility location problem has several applications in practice \cite{HererRH}. We show that the single-sink hard capacitated $k$-facility location problem, in which $D$ contains exactly one client, is NP-hard even when $f_i=0, i\in F$. We give the first FPTAS for SCKFL by extending the FPTAS for the knapsack problem. To the best of our knowledge, this is also the fist FPTAS for the single-sink capacitated facility location problem, which answers a question by G{\"o}rtz and Klose \cite{GortzK}.

(\romannumeral2) For the hard capacitated $k$-facility location problem with uniform capacities, in which $s_i=s, i\in F$,
we observe that for $|D|$=1, it is easy to find an optimal solution. A natural question is to extend this to any fixed number $m:=|D|$ of clients. We give a {polynomial time} algorithm for this setting that runs in time $O({\binom{n}{m}} \cdot n^3)$, where $n=|F|$. Using the structure of the graph consisting of {the fractional valued edges} in any extreme solution, the problem is reduced to a number of {transportation problems}.

(\romannumeral3) We observe that the number of fractionally open facilities can be bounded
by analyzing the rank of the constraint matrix corresponding to the tight constraints at a fractional extreme point solution.
Then, we give approximation algorithms for two variants of the hard capacitated $k$-facility location problem based on this upper bound.

Another example to show the potential power of the structure of extreme point solutions is that
we can slightly improve the previous best approximation ratio 5 obtained by Levi, Shmoys, and Swamy \cite{LeviSS}, and Bansal, Garg, and Gupta \cite{BansalGG} for the capacitated facility location problem with uniform opening costs, by combining our technique with a
pesudo-approximation algorithm for the $k$-median problem.

\section{The Single-sink Capacitated $k$-facility Location Problem}\label{sec:SCKF}
In this section, we consider the single-sink capacitated $k$-facility location problem (SCKF).
Since we only have one client with demand $d$, the formulation for the CKF is reduced to the following mixed integer program.

\begin{alignat}{2}
Z_{\mathrm{MIP}}=\min &\sum_{i\in F} {(c_i x_i+f_i y_i)} \label{MIP}\\
\text{subject to:\quad}&\sum_{i\in F}{x_{i}}  =  d, \label{demand constraint} \\
& \sum_{i\in F}{y_i}\leq k,\label{cardinality constraint} \\
&0\leq {x_{i}} \leq  s_i y_i,&&\forall i\in F,\label{capacity constraint} \\
&y_{i} \in \{0,1\},&&\forall i\in F.\label{binary constraint}
\end{alignat}

Again, the natural LP relaxation of SCKFL can be obtained by replacing constraints (\ref{binary constraint}) by (\ref{binary relaxation}).
\begin{lemma}
The single-sink capacitated $k$-facility location problem is NP-hard even when $f_i=0$ for all $i\in F$.
\end{lemma}

\begin{proof}
Consider the case that $s_i>1$, $c_i:=1-\frac{1}{s_i}$ and $f_i=0$ for all $i\in F$. We claim that
\begin{equation}\label{subsetsum}
Z_\mathrm{MIP}\leq d-k\iff \text{there exists $I\subseteq F$ with $|I|=k$ and $\sum_{i\in I}s_i=d.$}
\end{equation}
Indeed, for the objective value we find
\begin{equation*}
\sum_{i\in F}{c_{i} x_{i}}=d-\sum_{i\in F}{\frac{x_i}{s_i}}=d-\sum_{i\mid y_i=1}\frac{x_i}{s_i}\geq d-k,
\end{equation*}
where the last inequality holds because $x_i\leq s_i$ and $y_i=1$ for $k$ values of $i$.
Equality holds if and only if $x_i=s_i$ for all $i\in F$ with $y_i=1$.
That is, if and only if $\sum\{s_i\mid y_i=1\}=d$.

The claim above allows to reduce SUBSET-SUM to SCKFL as follows.
Let positive integers $s_1,\cdots, s_n>1$ and $d$ form an instance of SUBSET-SUM.
Now there {exists} a subset $I\subseteq \{1,2,\cdots,n\}$ such that $\sum_{i\in I}{s_i}=d$ if and only if the objective value of SCKFL is at most $d-k$ for some $k\in \{1,\cdots,n\}$.
\end{proof}

\begin{remark}\label{unbounded gap of SCKFL}
The integrality gap $Z_{\mathrm{MIP}}/Z_{\mathrm{LP}}$ is unbounded.
\end{remark}

Take the instance shown in Figure \ref{fig:instance} with four facilities $\{1,2,3,4\}$, $s_{1}=s_{2}=s, s_{3}=Ms, s_{4}=s+1$, $d=2s+1$, $f_{1}=f_{2}=f_{3}=f_{4}=0$, and $c_{1}=c_{2}=0, c_{3}=100, c_{4}=1, k=2$ and $M\gg s \gg 100$.
For this instance, we have $Z_{\mathrm{MIP}}=s+1$ and $Z_{\mathrm{LP}}=\frac{100M}{M-1}$.
Thus, $Z_{\mathrm{MIP}}/Z_{\mathrm{LP}}=\frac{s+1}{\frac{100M}{M-1}}>\frac{s+1}{200}$,
which can be arbitrarily large. In addition, a simple LP-rounding technique does not work for SCKFL.
For the above instance, an optimal solution for LP-relaxation is $y_{1}=1,y_2=\frac{Ms-s-1}{(M-1)s}, y_3=\frac{1}{(M-1)s}, x_1=s,x_2=\frac{Ms-s-1}{M-1}, x_3=\frac{M}{M-1}$.
A natural idea is to round $y_3$ to be $1$, $y_2$ to be $0$.
It is clear that the objective value of the solution obtained by this simple rounding is still really large.
\begin{figure}[!htb]
\centering
\includegraphics[totalheight=1.0in]{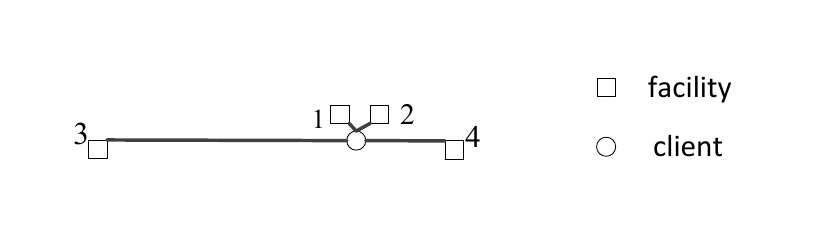}
\caption{\textsl{An instance for SCKFL. An optimal solution for the LP-relaxation of this instance is $y_{1}=1,y_2=\frac{Ms-s-1}{(M-1)s}, y_3=\frac{1}{(M-1)s}, x_1=s,x_2=\frac{Ms-s-1}{M-1}, x_3=\frac{M}{M-1}$, with the total cost $\frac{100M}{M-1}$. An optimal solution for the MIP is $y_1=y_4=1, x_{1}=s, x_{4}=s+1$ with the total cost $s+1$. }}
\label{fig:instance}
\end{figure}

We aim to design a fully polynomial time approximation scheme (FPTAS) for SCKFL.
Before introducing our algorithm, we present a key observation
(P{\'a}l, Tardos, and Wexler gave a similar observation in the proof of Lemma 3.3 in \cite{PalTW}).

\begin{observation}
For the single-sink capacitated $k$-facility location problem, there is an optimal solution $(x^*,y^*)$
in which at most one open facility $t^0$ is not fully used, i.e., $x^*_{i}\in\{0, s_{i}\}$ for $i\neq t^0$.
\end{observation}

Without loss of generality we suppose that $c_{ij}$ and $f_i$, for each $i\in F, j\in D$, are all integral.
Given $t^0$, which is allowed not to be fully used in an optimal integral solution $(x^*,y^*)$,
in order to solve SCKFL it is sufficient to solve the following problem for a given integer $p$:
\begin{equation}\label{goal}
\max \big\{\sum_{i\in F'}s_i \mid F'\subseteq F\setminus\{t^0\}, |F'|\leq k-1, \sum_{i\in F'}(c_is_i+f_i)=p\big\}.
\end{equation}
In words, we find for each total cost $p$ a set of at most $k-1$ facilities (not containing $t^0$) to open and use to full capacity, maximizing the total capacity.

We can recursively solve the above problem by dynamic programming. Without loss of generality, suppose $F\setminus \{t^0\}=\{1,2,\cdots,n-1\}$, where $n=|F|$. For nonnegative integers $p$ and $g\leq b\leq n-1$ define
\begin{equation*}
S_g(b,p):=\max \big\{\sum_{i\in F'}s_i\mid F'\subseteq \{1,\ldots, b\}, |F'|\leq g, \sum_{i\in F'}(c_is_i+f_i)=p\big\},
\end{equation*}
and let $F_g(b,p)$ be an optimal solution $F'$. If $\sum_{i\in F'}(c_is_i+f_i)=p$ does not hold for any $F'\subseteq \{1,\ldots, b\}$ with $|F'|\leq g$, we set $S_g(b,p):=-\infty$ and $F_g(b,p):=\emptyset$. Clearly, $S_g(0,0)=0$ and $S_g(0,p)=-\infty$ for $p>0$. The other values $S_g(b,p)$, and the corresponding optimum solutions $F_g(b,p)$, can be computed recursively since
\begin{equation*}
S_g(b+1,p)=\max \big(S_g(b,p),\ s_{b+1}+S_{g-1}(b,p-(f_{b+1}+c_{b+1} s_{b+1})\big)
\end{equation*}
for $0<g\leq b$. In the maximum, the two values correspond to not opening and opening facility $b+1$, respectively.

For computing the maximum in (\ref{goal}), it suffices to restrict to values $0\leq p\leq (k-1)\mathcal{P}\leq n\mathcal{P}$, where $\mathcal{P}=\max \{c_is_i+f_i\mid i\in \{1,\ldots, n-1\}\}$. Hence we can solve (\ref{goal}) in time $O(n^3 \mathcal{P})$.

Since $\mathcal{P}$ may be exponential in the size of the input of SCKFL, the computing time could be non-polynomial.
We overcome this difficulty by a scaling-and-rounding technique. The resulting Algorithm \ref{FPTAS} may be seen as a generalization of the FPTAS for the knapsack problem (with cardinality constraints) \cite{CapraraKPP,Lawler}.

\begin{assumption}
For each $i\in F, C_i>0$, where $C_i:=c_i s_i+f_i$.
\end{assumption}
Note that if $C_i=0$ and $s_i< d$ we directly open $i$ and serve demand $s_i$ of the single client by $i$ without increasing any cost. If $C_i=0$ and $s_i\geq  d$, the optimal total cost is 0.

\begin{algorithm}[!htlp]
\caption{An FPTAS for the single sink capacitated $k$-facility location problem}\label{FPTAS}
\begin{description}[leftmargin=*]
\item[Input] Finite set $F$ of facilities, costs $c\in \Z_{\geq 0}^F$, costs $f\in \Z_{\geq 0}^F$, demand $d\in \Z_{\geq 0}$, capacities $s\in \Z_{\geq 0}^F$, integer $1\leq k\leq n:=|F|$, $\epsilon>0$.
\item[Output] A feasible solution $(x,y)$ that is within a factor $1+\epsilon$ of optimum, if a feasible solution exists.
\item[Description]\mbox{}\\
\begin{itemize}
\item[1.] Order facilities such that $C_1\leq C_2\leq \cdots \leq C_n$, where $C_i:=c_i s_i+f_i$.
\item[2.] for $t=1$ to $n$ do
\item[]  \quad for $r=1$ to $n$ do
\item[]    \quad\quad if $\{1,\cdots,r\}-\{t\}=\emptyset$,
\item[]    \quad\quad\quad Let $C_{r}^0=0,$ $\bar{f}_t={f_t}$, $\bar{c}_t={c_t}$, $W=1$.
\item[]    \quad\quad end if
\item[] \quad\quad if $\{1,\cdots,r\}-\{t\}\neq \emptyset$
\item[]    \quad\quad\quad Let $W=\frac{\epsilon C_{r}^0 }{k}$, where $C_{r}^0=\max \{C_i \mid i\in\{1,\cdots,r\}-\{t\} \}$.
\item[]    \quad\quad\quad For each facility $i\in \{1,\cdots,r\}- \{t\}$, define $\bar{C}_i=\lfloor\frac{C_i}{W}\rfloor$.
\item[]    \quad\quad\quad Let $\bar{f}_t=\frac{f_t}{W}$, $\bar{c}_t=\frac{c_t}{W}$.
\item[]    \quad\quad end if
\item[]    \quad\quad Consider the subproblem $P_{rt}$ involving items $\{1,\cdots,r\}\cup \{t\}$, in which
\item[]      \quad\quad only $t$ can be not fully used, that is, $x_i\in \{0,s_i\}, i\in \{1,\cdots,r\}-\{t\}$;
\item[]      \quad\quad$0\leq x_t \leq s_t$. With the above scaled costs, compute $S_g(\bar{r},p)$ for each
\item[]      \quad\quad $0\leq g\leq k-1$, $0\leq p\leq (k-1) \lfloor \frac{C_{r}^0}{W} \rfloor$, where $\bar{r}=r$ if $r\neq t$, $\bar{r}=r-1$
\item[]      \quad\quad otherwise. Then, find a solution with total scaled cost:
\item[]     \quad\quad  $\min\{p+(d-S_g(\bar{r},p))\bar{c}_t+\bar{f}_t \mid 0\leq d-S_g(\bar{r},p)\leq s_t,0\leq g\leq k-1,$
\item[]      \quad\quad\quad\quad $0\leq p\leq (k-1)\lfloor \frac{C_{r}^0}{W} \rfloor \},$ if a feasible solution exists.
\item[]  \quad end for
\item[] end for
\item[3.] for $r=1$ to $n$ do
\item[]    \quad if $s_r\geq d$,
\item[]    \quad\quad find a solution with total cost $d c_r+f_r$.
\item[]    \quad end if
\item[] end for
\item[4.] Output the solution with the minimum total original cost.
\end{itemize}
\end{description}
\end{algorithm}
\begin{remark}
Note that in Algorithm \ref{FPTAS} for nonnegative integers $p$ and $g\leq b\leq \bar{r}$,
$$S_g(b,p):=\max \big\{\sum_{i\in F'}s_i\mid F'\subseteq \{1,\ldots, b\}-\{t\}, |F'|\leq g, \sum_{i\in F'}(c_is_i+f_i)=p\big\}.$$
\end{remark}
\begin{theorem}
Let $OPT$ be the cost of an optimal solution, $SOL$ be the cost of the solution returned by Algorithm \ref{FPTAS}.
Then, $SOL\leq (1+\epsilon)OPT$. The running time of Algorithm \ref{FPTAS} is $O(\frac{n^6}{\epsilon})$, for any $\epsilon>0$.
\end{theorem}

\begin{proof}
Suppose $(x^*,y^*)$ is an optimal solution in which at most one open facility is not fully used.
Let $t^0$ be the open facility in $(x^*,y^*)$ that is not fully used if it exists.
Otherwise, let $t^0$ be some open facility in $(x^*,y^*)$.
Then, we define $F^*_1=\{i\in F \mid y^*_i=1\}-\{t^0\}$ as the set of opened facilities in $(x^*,y^*)$ excluding $t^0.$

If $F^*_1=\emptyset$, clearly our algorithm can find an optimal solution in Step 3.
If $F^*_1\neq \emptyset$, let $C_{i}^0=\max \{C_i \mid i\neq t^0, y^*_i=1\}$. Note that $C_{i}^0\leq OPT$.
Moreover, let $i^0=\max \{i\in F^*_1 \mid C_i=C_{i}^0\}$. Thus, $C_{i^0}=C_{i}^0.$

Suppose in iteration $t=t^0,r=i^0$ of Step 2 we get an optimal solution $(x,y)$.
Let $F_1=\{i\in F \mid y_i=1\}-\{t^0\}.$
Let $Cost(x,y)$ and $Scaled\_cost(x,y)$ be the original and scaled total cost of solution $(x,y)$ respectively.
So, $Cost(x,y)=(\sum_{i\in F_1}{C_i})+x_{t^0} c_{t^0}+f_{t^0}$,
and $Scaled\_cost(x,y)=(\sum_{i\in F_1}{\bar{C}_i})+x_{t^0} \bar{c}_{t^0} +\bar{f}_{t^0},$
where the definition of $\bar{C}_i$ is given in Algorithm \ref{FPTAS}.
We will show that $Cost(x,y)\leq (1+\epsilon)OPT$, which then implies $SOL\leq (1+\epsilon)OPT$.

Recall that $W=\frac{\epsilon C_{r}^0}{k}$. We have
\begin{eqnarray*}
Cost(x,y)&&=(\sum_{i\in F_1}{C_i})+x_{t^0} c_{t^0}+f_{t^0}\\
&& \leq (\sum_{i\in F_1}{(W \bar{C}_i+W)})+W x_{t^0} \bar{c}_{t^0}+W \bar{f}_{t^0}\\
&& \leq W((\sum_{i\in F_1}{\bar{C}_i})+x_{t^0} \bar{c}_{t^0} +\bar{f}_{t^0})+kW\\
&& \leq {W} \cdot Scaled\_cost(x,y) + kW,
\end{eqnarray*}
where the second inequality holds as $|F_1|\leq k-1$.

The scaled total cost of solution $(x^*,y^*)$ in this iteration is
$(\sum_{i\in F^*_1}{\bar{C}_i})+x^*_{t^0} \bar{c}_{t^0}+\bar{f}_{t^0}$.
Clearly,
\[
Scaled\_cost(x,y)\leq (\sum_{i\in F^*_1}{\bar{C}_i})+x^*_{t^0} \bar{c}_{t^0}+\bar{f}_{t^0},
\] since $(x,y)$ is optimal in this iteration. That is,
\[
Scaled\_cost(x,y)\leq (\sum_{i\in F^*_1}{\lfloor\frac{C_i}{W}\rfloor})+x^*_{t^0} \frac{c_{t^0}}{W}+\frac{f_{t^0}}{W}.
\]

Then, we have
\[
{W}\cdot Scaled\_cost(x,y)\leq {W} (\sum_{i\in F^*_1}{\lfloor\frac{C_i}{W}\rfloor})
+{W} x^*_{t^0} \frac{c_{t^0}}{W}+{W}\frac{f_{t^0}}{W}
\]
\[
\Rightarrow {W} \cdot Scaled\_cost(x,y) + kW \leq (\sum_{i\in F^*_1}{C_i})+ x^*_{t^0} {c_{t^0}}+ {f_{t^0}}+kW.
\]

Therefore, we get
\[
Cost(x,y) \leq (\sum_{i\in F^*_1}{C_i})+ x^*_{t^0} {c_{t^0}}+ {f_{t^0}}+kW= OPT+\epsilon C_{i}^0\leq (1+\epsilon)OPT,
\]
where the equality holds by the definition of $W$ and the last inequality holds as $C_{i}^0\leq OPT$.

For fixed $t$, the running time of the subproblem $P_{rt}, r=1,\cdots,n$ is $O(n^3 \lfloor\frac{C_{r}^0}{W}\rfloor)$.
That is, $O(n^3 \frac{k}{\epsilon})$. Thus, the total running time of our algorithm is $O(\frac{n^6}{\epsilon})$
as we have $O(n^2)$ subproblems.
\end{proof}

\section{The Capacitated $k$-facility Location Problem with Uniform Capacities}\label{sec:CKFLU}
In this section, we aim to show the following result for the capacitated $k$-facility location problem with uniform capacities (CKFU). Let $m=|D|$, $n=|F|$ and $s_i=s, i\in F$.
\begin{theorem}\label{CKP with uniform capacity}
For fixed $m$, the capacitated $k$-facility location problem with uniform capacities can be solved in polynomial time $O({\binom{n}{m}} \cdot n^3)$.
\end{theorem}

We need new notation to describe our idea. We consider an optimal solution $(x,y)$ for CKFLU as a weighted bipartite graph $G=(V,E)$, where $V=\{i\in F \mid y_i=1\}\cup D$ and $E=\{\{i,j\} \mid x_{ij}>0,i\in F, j\in D\}$. To be more precise, if $x_{ij}>0$, we add an edge $\{i,j\}$ between facility $i$ and client $j$ with weight $x_{ij}$. Moreover, let $\bar{E}=\{\{i,j\}\in E \mid 0< x_{ij}<s\}$ and $\bar{V}=\bigcup_{e\in \bar{E}}{e}$. We call $(\bar{V},\bar{E})$ the \emph{untight} weighted subgraph of $G$.

Define $r_j:=d_j/s$ for all $j \in D$. If all $r_j$ are integral, we say that the CKFLU is \emph{divisible}.

\begin{lemma}\label{demand is multiple of capacity}
The divisible capacitated $k$-facility location problem with uniform capacities can be solved in $O(n^3)$ time.
\end{lemma}

\begin{proof}
We transform the divisible CKFLU to a balanced transportation problem, in which the total capacity is equal to total demand.
Then, to get an integer solution to this transportation problem, we can consider this problem as a minimum weight perfect matching problem that can be solved in $O(n^3)$ time \cite{Gabow}, by splitting the demands.
Since the problem is infeasible if $k< \sum_{j\in D}{r_j}$, we only consider the case: $|F|\geq k \geq \sum_{j\in D}{r_j}$.

By dividing the capacity and demand constraints by $s$, we can get an equivalent formulation for the divisible CKFLU, in which the new capacity of each facility is $1$ and the new demand of each client $j$ is $r_j$.

First, we show that there exists an optimal integral solution for this equivalent formulation. We add a dummy client $j^\prime$ to $D$ with demand $r_{j^\prime}=n-\sum_{j\in D}{r_j}$.
Take the cost of shipping one unit from $i\in F$ to $j\in D\setminus \{j'\}$ to be $s c_{ij}+ f_i$, from $i\in F$ to $j'$ to be $0$.
Now the divisible CKFLU can be considered as a balanced transportation problem with total demand $n$. Since $r_j, j\in D$ are integers, there is an integer optimal solution for this transportation problem (see for instance \cite{HoffmanK}, or Theorem 21.14 in \cite{Schrijver}). Note that based on the optimal integer solution for this transportation problem, we can easily construct an optimal solution for our original problem.

Then, to get an optimal integer solution for the constructed transportation problem, we can split each $j\in D$ to $r_{j}$ copies each with demand $1$. Now we can consider the balanced transportation problem as a minimum weight perfect matching problem that can be solved in $O(n^3)$ time\cite{Gabow}.
\end{proof}

Note that if we know the exact structure of $(\bar{V},\bar{E})$, then according to the definition of $G$ the remaining part $(V,E\setminus \bar{E})$ can be generated by an optimal integer solution to an instance of the divisible CKFLU problem. Thus, the high-level idea is that we reduce our original problem to a collection of divisible CKFLU problems by checking all the possible structures of $(\bar{V},\bar{E})$. {To prove that we can examine all the structures in polynomial time, we show some useful properties of the untight weighted subgraph of $G$ first.}

\begin{lemma}\label{properties for untight graph}
Let $G=(V,E)$ be the graph corresponding to a vertex $(x,y)$ of the convex hull of feasible solutions of the MIP to CKFLU, and $H=(\bar{V},\bar{E})$ be its corresponding untight subgraph. Then,
\begin{itemize}
\item[(a)] $G$ is acyclic;
\item[(b)] in each connected component of $H$, there is at most one $i\in F\cap \bar{V}$ with $0<\sum_{j\in D}{x_{ij}}<s$;
\item[(c)] $H$ contains at most $m$ facilities and $2m-1$ edges.
\end{itemize}
\end{lemma}

\begin{proof}
(a). Suppose that there is a cycle $O=(e_1,e_2,\ldots, e_{2p-1},e_{2p})$ in $G$. Note that $O$ must have even number of edges as $G$ is bipartite. Let $\chi^O\in \mathbb{R}^E$ be the signed incidence vector of this path:
\begin{eqnarray*}
&&\chi^O(e_i)=(-1)^i\text{ for $i=1,3\ldots,2p-1$}; \chi^O(e_i)=(1)^i\text{ for $i=2,4\ldots,2p$};\\
&&\chi^O(e)=0\text{ for } e\in E\setminus\{e_1,\ldots,e_{2p}\}.
\end{eqnarray*}

For sufficiently small $\epsilon >0$ both $(x+\epsilon \chi^O,y)$ and $(x-\epsilon \chi^O,y)$ are feasible solutions, contradicting the fact that $(x,y)$ is a vertex.

(b). The idea is similar to (a). Consider any connected component $B$ of $H$. Suppose for contradiction that we have two facilities $i_1,i_2$ in $B$ with $0<\sum_{j\in D}{x_{i_1 j}}<s, 0<\sum_{j\in D}{x_{i_2 j}}<s$. Since $B$ is connected, there is a path $P=(e_1,e_2,\ldots, e_{2p-1},e_{2p})$ from $i_1$ to $i_2$. Again, we can construct two feasible solutions $(x+\epsilon \chi^P,y)$ and $(x-\epsilon \chi^P,y)$, contradicting the fact that $(x,y)$ is a vertex.

(c). Consider any connected component of $H$ with at least one edge. Note that each component is a tree with $0<x_{ij}<s$ for each edge $\{i,j\}$. If there is a facility $i^*$ in this component with $0<\sum_{j\in D}{x_{i^* j}}<s$, then take $i^*$ as the root. Otherwise, take an arbitrary facility as the root. Since $0<x_{ij}<s$ for each edge $\{i,j\}$ and $\sum_{j\in D}{x_{ij}}=s$ for each facility $i\neq i^*$, each facility except $i^*$ has at least two neighbors. Then, each facility in this connected component has at least one child (client) as each facility has at most one parent. Moreover, no two facilities have a common child (by the definition of a rooted tree). Therefore, the number of facilities in each connected component is at most the number of clients. Thus, we have at most $m$ facilities in $H$ as there are at most $m$ clients. Clearly, the number of edges is at most $2m-1$ since $H$ is a forest.
\end{proof}

\begin{lemma}\label{unique weight for untight graph}
For any untight and acyclic subgraph $H=(\bar{V},\bar{E})$, given the set $I=\{i\in F\cap \bar{V} \mid 0<\sum_{j\in D}{x_{ij}}<s \}$, we can get the unique weight $x_{ij}$ for each edge $\{i,j\}\in \bar{E}$ in $O(m)$ time.
\end{lemma}

\begin{proof}
Consider any connected component of $H$. Note that each connected component must be in the form of a tree.
If there is a facility $i^*\in I$ in this component, then take $i^*$ as the root.
Otherwise, take an arbitrary facility $i^*$ in this component as the root.
Then, all leaves are clients since $\sum_{j\in D}{x_{ij}}=s$ for each facility $i\neq i^*$ in the considered connected component
 (Lemma \ref{properties for untight graph}(b)) and $0<x_{ij}<s$ for each edge $\{i,j\}$.

We will show that in each connected component, if node (client) $j$ is a leaf, we can obtain the exact value of $x_{ij}$, where $i$ is the father of $j$;
and for each other node in this tree, we can compute the value of the edge between this node and its father
based on the values of edges between this node and its children.
Then, we can obtain the values of all edges in the tree by induction.

Consider a client $j$. Let $f(j)$ be the father node (facility) of $j$ in the tree and $c(j)$ be the set of children (facilities) of $j$.
If $j$ is a leaf, that is $c(j)=\emptyset$, then we know $|\{i\in F \mid 0<x_{ij}<s\}|=1$. Otherwise, $j$ cannot be a leaf.
Thus, we can get the exact value for $x_{f(j),j}=d_j-\lfloor \frac{d_j}{s} \rfloor \cdot s$
since $j$ has exactly one father.
If $j$ is not a leaf, the value $x_{f(j),j}=(d_j-\sum_{i\in c(j)}{x_{ij}})-\lfloor \frac{d_j-\sum_{i\in c(j)}{x_{ij}}}{s} \rfloor \cdot s$ as $x_{tj}\in \{0,s\}, \forall t\in V\setminus \bar{V}$.

Consider a facility $i\neq i^*$. Let $f(i)$ be the father node (client) of $i$ in the tree and $c(i)$ be the set of children (clients) of $i$. We can obtain the value of $x_{i,f(i)}$ as long as all values of $x_{ij}, j \in c(i)$ are known, since $i$ must be fully used by Lemma \ref{properties for untight graph}. That is, $x_{i,f(i)}=s-\sum_{j\in c(i)}{x_{ij}}$. Note that if $i=i^*$, we can stop since $f(i^*)=\emptyset$.

Moreover, the computing time is $O(m)$ since each edge is only examined once.
\end{proof}

Consider an optimal integer vertex $(x,y)$ of the convex hull of feasible solutions for CKFLU whose corresponding graph $G=(V,E)$ is a forest.
The graph $H=(\bar{V},\bar{E})$ (the untight subgraph of $G$) can be viewed as a subgraph of some spanning tree of the complete bipartite graph $K_{\bar{F},D}$, where $\bar{F}=F\cap \bar{V}$. Consequently, checking all the possible structures of $H$ means checking all the subgraphs of these spanning trees. Note that $H$ and $K_{\bar{F},D}$ have the same vertices. Then, it now suffices to answer the following questions:
\begin{itemize}
\item[1.] how many different complete bipartite graphs do we have for $K_{\bar{F},D}$?
\item[2.] how to list all the spanning trees for a complete bipartite graph?
\item[3.] how many subgraphs, that have the same vertices as the considered spanning tree, does a spanning tree have?
\item[4.] for a fixed structure of $H$, how to compute the corresponding total cost?
\end{itemize}
If all the above questions can be solved in polynomial time, we can get all the possibilities of $H$ in polynomial time. Consequently, Theorem \ref{CKP with uniform capacity} can be proved by Lemma \ref{demand is multiple of capacity} and \ref{unique weight for untight graph}.

\paragraph{Proof of Theorem \ref{CKP with uniform capacity}.} Because $H=(\bar{V},\bar{E})$ contains at most $m$ facilities by Lemma \ref{properties for untight graph}, the number of all the possible cases for $K_{\bar{F},{D}}$ can be bounded by $\sum_{t=1}^{m}{\binom{n}{t}}\leq m \cdot {\binom{n}{m}}$. So, we can answer question 1.

Lemma \ref{list all spanning tree} and \ref{the number of spanning trees for bipartite graph} answer question 2. The time to list all the spanning trees for the complete bipartite graph is $O(m^{2m-2}+2m+m^2)$ since we have at most $m$ facilities and $m$ clients in $K_{\bar{F},{D}}$ by Lemma \ref{properties for untight graph}. Note that at this stage, we do not need to consider the weight $x_{ij}$ of edge $\{i,j\}$.

By Lemma \ref{properties for untight graph}, we know that the number of edges is at most $2m-1$ in a spanning tree. Thus, each spanning tree has at most $2^{2m-1}$ subgraphs that have the same vertices as the spanning tree. This answers question 3.

Then, the total time to list all the possible untight subgraphs is $O( m \cdot {\binom{n}{m}} \cdot (m^{2m-2}+2m+m^2) \cdot 2^{2m-1})$.

By Lemma \ref{unique weight for untight graph}, we can get the cost for any untight subgraph in polynomial time $O(m)$ as long as $I=\{i\in F\cap \bar{V} \mid 0<\sum_{j\in D}{x_{ij}}<s \}$ is fixed. {Note that the opening costs for facilities are easy to get if we know the structure of $H$. Indeed, it is $\sum_{i\in F\cap \bar{V}}{f_i}$}.  The remaining part $(V,E\setminus \bar{E})$ can be considered as an optimal integer solution to a divisible CKFLU, which means we can get the total cost in polynomial time $O(n^3)+O(m)$ by Lemma \ref{demand is multiple of capacity}. This answers question 4. Moreover, the number of all the choices for $I$ is bounded by $2^m$ since there are at most $m$ facilities in each spanning tree by Lemma \ref{properties for untight graph}.

Combining all the pieces together, we can get all the possibilities of solutions in computing time $O( m \cdot {\binom{n}{m}} \cdot (m^{2m-2}+2m+m^2) \cdot 2^{2m-1} \cdot 2^m \cdot (m + n^3))= O({\binom{n}{m}} \cdot (m^{2m-1}+2 m^2+m^3) \cdot 2^{3m-1} \cdot (m+n^3))$, that is, $O({\binom{n}{m}} \cdot n^3)$. Finally, we output the solution with at most $k$ open facilities and the smallest total cost.\qed

\begin{lemma}\label{list all spanning tree}
\emph{\cite{KapoorR}} For an undirected graph without weight $G=({V},{E})$, all spanning trees can be correctly generated in $O(N+|{V}|+|{E}|)$ time, where $N$ is the number of spanning trees.
\end{lemma}

\begin{lemma}\label{the number of spanning trees for bipartite graph}
\emph{\cite{Scoins}} The number of spanning trees of a complete bipartite graph is $m^{n-1} n^{m-1}$, where $m$ and $n$ are respectively the cardinalities of two disjoint sets in this bipartite graph.
\end{lemma}

\section{The Hard Capacitated $k$-facility Location Problem with Non-uniform Capacities}
In this section, we show how to bound the number of fractionally open facilities by a simple rank-counting argument on an extreme point solution. Then, together with an algorithm to group clients, we give a simple constant factor approximation algorithm for the hard capacitated $k$-facility location problem with non-uniform capacities (CKFL) (with uniform opening cost) with approximation ratio $7+\epsilon$ by using at most $2k$ facilities. As a simple illustration of the techniques used, we first give a 2-approximation algorithm for the single-sink hard capacitated $k$-facility location problem (SCKFL). Note that this ratio is worse than that of the FPTAS in Section \ref{sec:SCKF}. Here we aim to show that this upper bound is helpful to design approximation algorithms. And the approach is totally different from the FPTAS.

\subsection{A Simple Illustration of Using the Structure of Extreme Point Solutions}
\noindent \textbf{The Structure of Extreme Point Solutions to SCKFL}

\begin{definition}
Let $Ax\leq a, Bx\geq b, Cx=c$ be a system of linear (in)equalities. For a feasible solution $z$ we define \emph{the rank at $z$} of the system to be the (row)rank of $\begin{bmatrix}A_z^\transp&B_z^\transp&C^\transp\end{bmatrix}^\transp$, where $A_zx\leq a_z, B_zx\geq b_z, Cx=c$ is the subsystem consisting of the (in)equalities that are satisfied with equality by $z$.
\end{definition}
Note that for two subsystems, the sum of the ranks at $z$ of those two subsystems is at least the rank at $z$ of their union.

Let $P$ be the set of feasible solutions to the system SCKFL-LP consisting of (\ref{binary relaxation}), (\ref{demand constraint}),(\ref{capacity constraint}) and $\sum_{i\in F}{y_i}=k$
(Note that in this section we consider constraint $\sum_{i\in F}{y_i}=k$ instead of the corresponding inequality (\ref{cardinality constraint})). That is,
$$P:=\{(x,y) : \text{SCKFL-LP}\},$$ where SCKFL-LP is a system of constraints given below:
\begin{alignat}{2}
&\sum_{i\in F}{x_{i}}  =  d,\quad \sum_{i\in F}{y_i}&& =  k,\label{SCKF-LP}\\
&0\leq {x_{i}} \leq  s_i y_i,&&\forall i\in F,\nonumber\\
&0\leq y_i\leq 1,&&\forall i\in F.\nonumber
\end{alignat}

\begin{lemma}\label{twofractional}
Let $(x,y)$ be a vertex of $P$. Then either $y$ is integer, or $y$ has exactly two noninteger components and for every $i\in F$ we have $x_i=0$ or $x_i=s_iy_i$.
\end{lemma}
\begin{proof}
Let $F':=\{i\in F\mid 0<y_i<1\}$. If $|F'|=0$ we are done. As $|F'|=1$ is ruled out because the sum of the $y_i$ is $k, k\in \Z$, we may assume that $|F'|\geq 2$.

The rank of system SCKFL-LP at $(x,y)$ is equal to $2n, n=|F|$(Theorem 5.7 in \cite{Schrijver}). We partition the (in)equalities in this system and bound the rank at $(x,y)$ for each subsystem:
\begin{itemize}
\item The rank at $(x,y)$ of the subsystem $\sum_{i\in F} x_i=d,\sum_{i\in F} y_i=k$ is at most $2$.
\item For every $i\in F'$, the rank at $(x,y)$ of the subsystem $0\leq x_i, x_i\leq s_iy_i, 0\leq y_i, y_i\leq 1$ is at most $1$ and equality holds if and only if $x_i=0$ or $x_i=s_iy_i$.
\item For every $i\in F\setminus F'$, the rank at $(x,y)$ of the subsystem $0\leq x_i, x_i\leq s_iy_i, 0\leq y_i, y_i\leq 1$ is at most $2$ and equality holds if and only if $x_i=0$ or $x_i=s_iy_i$.
\end{itemize}
Since the rank is subadditive, we find that the rank at $(x,y)$ of SCKFL-LP is at most
\begin{equation}
2+|F'|+2|F\setminus F'|=2n+2-|F'|\leq 2n,\nonumber
\end{equation}
where the inequality holds as $|F'|\geq 2,$
with equality only if $|F'|=2$ and for each $i$ we have $x_i=0$ or $x_i=s_i y_i$.
\end{proof}

We give a 2-approximation algorithms for SCKFL to show the potential power of this nice structure.

\vspace*{8pt}
\noindent \textbf{2-Approximation Algorithm for SCKFL}

\vspace*{8pt}
We give an alternative approach to get an approximate solution for SCKFL, compared to the FPTAS in Section \ref{sec:SCKF}.
This approach can be viewed as incomplete implement of a branch and bound technique, branching on the 0-1 variables $y_i$.
To obtain a $2$-approximation algorithm that runs in polynomial time, we use two key ideas.
First, by Lemma \ref{twofractional}, we know in any vertex of the feasible region of the LP-relaxation that either 0 or 2 components of $y$ are fractional. We exploit this to guide the branching.
Secondly, we show that for a branch $y_i=1$ either there is no 2-approximation solution, or we can find a 2-approximation solution in polynomial time by again exploiting the structure of the vertices of the feasible region to the LP-relaxation.
A precise description of this algorithm is given in Algorithm \ref{2-approx}.

\begin{algorithm}[!htbp]
\caption{A $2$-approximation algorithm for the single-sink hard capacitated $k$-facility location problem}\label{2-approx}
\begin{description}[leftmargin=*]
\item[Input] Finite set $F$ of facilities, costs $c\in \Z_{\geq 0}^F$, costs $f\in \Z_{\geq 0}^F$, capacities $s\in \Z_{\geq 1}^F$, demand $d\in \Z_{\geq 1}$, integer $k\in \Z_{\geq 1}$.
\item[Output] A feasible solution $(x,y)$ to MIP: (\ref{MIP}),(\ref{capacity constraint}),(\ref{binary constraint}), and (\ref{SCKF-LP}),
 that is within a factor $2$ of optimum, if a feasible solution exists.
\item[Description]\mbox{}\\
\begin{itemize}
\item[1.] Find an optimal vertex $(x,y)$ of the feasible region of the LP-relaxation.\\
If no solution exists then stop. If $y$ is integer then return $(x,y)$ and stop.
\item[2.] Let $i_1\neq i_2$ in $F$ with $y_{i_1},y_{i_2}\in (0,1)$ and $s_{i_1}\geq s_{i_2}$.
\item[3.] Define $x^1$ by $x^1_{i_1}:=x_{i_1}+x_{i_2}$, $x^1_{i_2}:=0$ and $x^1_i:=x_i$ for $i\neq i_1,i_2$.\\
Define $y^1$ by $y^1_{i_1}:=1$, $y^1_{i_2}:=0$, $y^1_i:=y_i$ for $i\neq i_1,i_2$.
\item[4.] Recursively compute a 2-approximation solution $(x^0,y^0)$ for the restriction to $F\setminus\{i_1\}$ and extend it by setting $x^0_{i_1}:=0$ and $y^0_{i_1}:=0$.\\
\item[5.] Set $F_0:=\emptyset$. While $|F_0|\leq |F|-k$ do:
\begin{itemize}
\item[a.] Find an optimal vertex $(x',y')$ of the feasible region of the LP-relaxation intersected with $\{(x,y)\mid y_{i_1}=1, y_i=0\ \forall i\in F_0\}$.
\item[b.] If $y'$ is integer, return the best solution among $(x',y')$, $(x^0,y^0)$ and $(x^1,y^1)$ and stop.
\item[c.] If $x'_{i_1}=s_{i_1}$, return the best solution among $(x^0,y^0)$ and $(x^1,y^1)$ and stop.
\item[d.] Let $i_3\neq i_4$ in $F$ with $y'_{i_3},y'_{i_4}\in (0,1)$ and $f_{i_3}\leq f_{i_4}$.
\item[e.] Define $y''$ by $y''_{i_1}:=0$, $y''_{i_3}:=y''_{i_4}:=1$ and $y''_i:=y'_i$ for $i\neq i_1,i_3,i_4$.\\
If $(x',y'')$ has smaller value than $(x^0,y^0)$, set $(x^0,y^0)\leftarrow (x',y'')$.
\item[f.] Set $F_0\leftarrow F_0\cup \{i_4\}$.
\end{itemize}
\end{itemize}
\end{description}
\end{algorithm}

\begin{theorem}
For the single-sink hard capacitated $k$-facility location problem, Algorithm \ref{2-approx} finds a solution that is within a factor $2$ of optimum, or it concludes correctly that there is no feasible solution. The running time is polynomially bounded in the number $|F|$ of facilities.
\end{theorem}

\begin{proof}
Notice that an optimal vertex of the feasible region of the LP-relaxation can be found in polynomial time (see for instance \cite{GrotschelLS}).
Furthermore, since the number of recursive calls is no more than $|F|-1$, the polynomial running time is evident. It now suffices to show that when the MIP: (\ref{MIP}),(\ref{capacity constraint}),(\ref{binary constraint}), and (\ref{SCKF-LP}) is feasible, the solution given by Algorithm \ref{2-approx} is within a factor two of optimum.

Clearly, if $y$ is integer in Step 1 of Algorithm \ref{2-approx}, then the output $(x,y)$ is an optimal feasible solution. Hence, by Lemma \ref{twofractional}, we may assume that $y$ has exactly two fractional components $y_{i_1}$ and $y_{i_2}$.
Then, we know $y_{i_1}+y_{i_2}=1$ since $\sum_{i\in F} y_i=k$, and all $y_i, i\in F$ are integer except $y_{i_1}$ and $y_{i_2}$.
Without loss of generality we can assume that $s_{i_1}\geq s_{i_2}$.

To see that $(x^1,y^1)$ defined in Step 3 of Algorithm \ref{2-approx} is indeed a feasible solution, it suffices to show that $x^1_{i_1}\leq s_{i_1}$. This follows directly from the fact that $s_{i_1}\geq s_{i_2}$, since
\begin{eqnarray*}
x^1_{i_1}=x_{i_1}+x_{i_2}&\leq& y_{i_1}s_{i_1}+y_{i_2}s_{i_2}\leq y_{i_1}s_{i_1}+y_{i_2}s_{i_1} =s_{i_1}.
\end{eqnarray*}
Further, we find an upper bound for the value of $(x^1,y^1)$,
\begin{equation}\label{xhatyhat}
c^\transp x^1+f^\transp y^1\leq (c^\transp x+f^\transp y)+(c_{i_1}s_{i_1}+f_{i_1}),
\end{equation}
which is at most the optimum plus $c_{i_1}s_{i_1}+f_{i_1}$.

To conclude the proof, we analyse Step 5 of Algorithm \ref{2-approx}. Observe that the initial solution $(x^0,y^0)$ may be replaced, but only by a better solution. Also observe, that the solution that is returned is always at least as good as $(x^0,y^0)$ and $(x^1,y^1)$. Hence, we may assume that $(x^0,y^0)$ (at the end of the algorithm) and $(x^1,y^1)$ are not $2$-approximations.
Let $(x^*,y^*)$ be an optimal solution. We have $y^*_{i_1}=1$, since otherwise $(x^0,y^0)$ would be a 2-approximation already at Step 4. It suffices to show that $(x^*,y^*)$ remains feasible throughout the iterations of Step 5, until a solution of the same value is returned in Step 5b. For this, we observe that while $(x^*,y^*)$ is feasible, the situation $x'_{i_1}=s_{i_1}$ as in Step 5c cannot occur, because otherwise, by (\ref{xhatyhat}), we would have
\begin{equation*}
c^\transp x^1+f^\transp y^1\leq c^\transp x+f^\transp y+(c_{i_1}s_{i_1}+f_{i_1})\leq c^\transp x+f^\transp y+ c^\transp x'+f^\transp y'\leq 2(c^\transp x^*+f^\transp y^*),
\end{equation*}
contradicting the fact that $(x^1,y^1)$ is not a $2$-approximation.

In Step 5d, the fact that $y'$ has exactly two fractional components follows from Lemma \ref{twofractional} as $y'$ is a vertex of a face of the feasible region of SCKFL-LP, and hence of that region itself. Observe that this implies that $y'_{i_3}+y'_{i_4}=1$, hence $(x',y'')$ defined in Step 5e is a feasible solution.

In Step 5f, we have $y^*_{i_4}=0$. Indeed, for the cost of $(x',y'')$ we find:
\begin{eqnarray*}
c^\transp x'+f^\transp y''&=&(c^\transp x'+f^\transp y')-f_{i_1}+(1-y'_{i_3})f_{i_3}+(1-y'_{i_4})f_{i_4}\\
&\leq& (c^\transp x'+f^\transp y')+f_{i_4} \leq (c^\transp x^*+f^\transp y^*)+f_{i_4}.
\end{eqnarray*}
Since $(x^0,y^0)$ and hence $(x',y'')$ is not a $2$-approximation, we find that $f_{i_4}> c^\transp x^*+f^\transp y^*$ and hence $y^*_{i_4}=0$. This shows that $(x^*,y^*)$ remains feasible after adding $i_4$ to $F_0$.
\end{proof}

\subsection{An Approximation Algorithm for CKFL with Uniform Opening Costs}\label{sec:CKFLNU_uniform_opening_costs}
In this section, we consider the capacitated $k$-facility location problem with uniform opening costs, i.e., $f_i=f, i\in F$. Since we have an upper bound on the number of fractionally open facilities based on Lemma \ref{CKPfractional} below, a {natural} idea is to design a constant factor approximation algorithm for CKFL by relaxing the cardinality constraint with a constant factor.
We give a simple algorithmic framework that can extend any $\alpha$-approximation algorithm for the (uncapacitated) $k$-median problem (UKM) to a $(1+2\alpha)$-approximation algorithm for CKFL using at most $2k$ facilities ($2k-1$ for uniform capacities).

The (uncapacitated) $k$-median problem (UKM) can be formulated as follows, where variable $x_{ij}$ indicates the fraction of the demand of client $j$ that is served by facility $i$, and $y_{i}$ indicates whether facility $i$ is open.
\begin{alignat*}{2}
\min\ &\sum_{i\in F}{\sum_{j\in D}{d_j c_{ij} x_{ij}}}\\
\text{subject to:\quad}&\sum_{i\in F}{x_{ij}}  =  1, &&\forall j\in D,\\
&x_{ij} \leq  y_i, &&\forall i\in F,\forall j\in D,\\
&\sum_{i\in F} y_i \leq  k,\\
&x_{ij}, y_{i} \in \{0,1\}, &&\forall i\in F,\forall j\in D.
\end{alignat*}

\vspace*{8pt}
\noindent \textbf{The Structure of Extreme Point Solutions to CKFL}

\vspace*{8pt}
Let $Q$ be the set of feasible solutions $(x,y)$ to the system CKFL-LP consisting of (\ref{eq:demand constraint}), (\ref{eq:capacity constraint}), (\ref{eq:cardinality constraint}), (\ref{eq:nonegative constraint}) and (\ref{binary relaxation}). That is, $$Q:=\{(x,y) : \text{CKFL-LP}\},$$ where CKFL-LP is a system of constraints given below:
\begin{align}
&\sum_{i\in F}{x_{ij}}=d_j,\quad \forall j\in D; \quad \sum_{i\in F} y_i\leq k, \nonumber \\
&\sum_{j\in D}{x_{ij}} \leq s_iy_i,\quad \forall i\in F,\nonumber\\
& x_{ij}\geq 0, \quad \forall i\in F, \forall j\in D,\nonumber\\
&0\leq y_i\leq 1,\quad \forall i\in F \nonumber.
\end{align}

\begin{lemma}\label{CKPfractional}
Let $(x,y)$ be a vertex of $Q$. Then $y$ has at most $m+1$ noninteger components, where $m=|D|$.
\end{lemma}

\begin{proof}
The proof is similar to the proof of Lemma \ref{twofractional}. Let $F'=\{i\in F \mid 0<y_i<1\}$. The rank of system CKFL-LP at $(x,y)$ is equal to $(m+1)n, n=|F|,m=|D|$. We partition the (in)equalities in this system and bound the rank at $(x,y)$ for each subsystem:
\begin{itemize}
\item The rank at $(x,y)$ of the subsystem $\sum_{i\in F}{x_{ij}}=d_j, \forall j\in D;\sum_{i\in F} y_i\leq k$ is at most $m+1$.
\item For every $i\in F'$, the rank at $(x,y)$ of the subsystem $\sum_{j\in D}{x_{ij}} \leq s_iy_i; 0\leq x_{ij}, j\in D; 0\leq y_i; y_i\leq 1$ is at most $m$ and equality holds if and only if $x_{ij}=0$ or $x_{ij}=s_i y_i$ for each $x_{ij}$.
\item For every $i\in F\setminus F'$, the rank at $(x,y)$ of the subsystem $\sum_{j\in D}{x_{ij}} \leq s_iy_i; 0\leq x_{ij},j\in D; 0\leq y_i; y_i\leq 1$ is at most $m+1$ and equality holds if and only if $x_{ij}=0$ or $x_{ij}=s_i y_i$ for each $x_{ij}$.
\end{itemize}
Since the rank is subadditive, we find that the rank of CKFL-LP is at most
\begin{equation}
m+1+m|F'|+(m+1)|F\setminus F'|=m+1+(m+1)n-|F'|.\nonumber
\end{equation}
So, we have $|F'|\leq m+1$ as $m+1+(m+1)n-|F'|\geq (m+1)n$.
\end{proof}

For the uniform capacities case ($s_i=s>0, \forall i\in F$),
we will show a stronger property that there is an optimal solution $(x,y)$ to the LP-relaxation
with at most $m$ noninteger components in $y$.
Indeed, consider an optimal solution $(x,y)$ with $|\{i\mid 0<y_i<1\}|$ minimal. Suppose for contradiction that $y$ has more than $m$ fractional components. Then there exist a client $j$ and two facilities $i_1, i_2$ such that $y_{i_1}$ and $y_{i_2}$ are fractional and $x_{i_1 j}, x_{i_2 j}>0$, and $x_{i_1 j}=s y_{i_1}, x_{i_2 j}=s y_{i_2}$ by Lemma \ref{CKPfractional}. Without loss of generality assume that $c_{i_1 j}\leq c_{i_2 j}$. Let $\epsilon:=\min\{sy_{i_2},s(1-y_{i_1})\}$. Now modify $(x,y)$ by setting
\begin{eqnarray*}
x_{i_1 j}:=x_{i_1 j}+\epsilon&&y_{i_1}:=y_{i_1}+\epsilon/s\\
x_{i_2 j}:=x_{i_2 j}-\epsilon&&y_{i_2}:=y_{i_2}-\epsilon/s,
\end{eqnarray*}
to obtain a new optimal solution, while $|\{i\mid 0<y_i<1\}|$ decreases, a contradiction. Thus, we can find an optimal solution $(x,y)$ to the LP-relaxation for which $y$ has at most $m$ noninteger components.

\vspace*{8pt}
\noindent \textbf{The Algorithm}

\vspace*{8pt}
We convert our original instance to a new instance with at most $k$ clients while incurring some bounded extra costs. Then, at most $2k$ facilities are (fractionally or fully) opened for the new instance according to Lemma \ref{CKPfractional}.

\begin{algorithm}[!htbp]
\caption{A $(1+2\alpha)$-approximation algorithm for CKFL with uniform opening costs using at most $2k$ facilities.}\label{CKFL-approx}
\begin{description}[leftmargin=*]
\item[Input] Finite set $F$ of facilities, $D$ of clients, costs $c\in \Q_{\geq 0}^{F\times D}$, opening cost $f\in \Q_{\geq 0}$, capacities $s\in \Q_{\geq 0}^F$, demands $d\in \Q_{\geq 0}^D$, integer $k\in \Z_{\geq 1}$.
\item[Output] A solution $(x,y)$ to MIP (\ref{eq:MIP})-(\ref{eq:binary constraint}) using at most $2k$ facilities
 that is within a factor $1+2\alpha$ of optimum, if a feasible solution exists.
\item[Description]\mbox{}\\

Suppose exactly $l$ facilities are opened in an optimal solution. That is,
we can consider a stronger constraint $\sum_{i\in F} y_i\leq l.$

\textbf{Step 1}. Reduce the input instance $I_0$ of CKFL to an instance $I_1$ of UKM as follows.\\
Let $F$ and $D$ be the set of facilities and clients of our input instance $I_0$ respectively.
Let $F^\prime=F$ (located at the same sites) be the set of facilities of UKM while with infinite capacities and without opening costs.
Let $D^\prime=D$ be the set of clients of UKM.
Solve this constructed instance (denoted by $I_1$) by the existing $\alpha$-approximation algorithm for UKM.
Suppose we get an integer solution $(x^\prime,y^\prime)$. Note that for UKM, there is an optimal solution
in a form of so-called stars.
That is, each client is served by exactly one open facility.
Without loss of generality, suppose $y^\prime_1=\cdots=y^\prime_l=1$.
Then, we can consider $(x^\prime,y^\prime)$ as $l$ stars $\{T_1,\cdots, T_l\}$,
where $T_r=\{j \in D^\prime \mid x^\prime_{r j}=1\}$ and the center of $T_r$ is the facility $r$.

\textbf{Step 2}. Consolidate clients and construct a new instance $I_2$ of CKFL with at most $l$ clients as follows.\\
For each star $T_r$ in $(x^\prime,y^\prime)$, we set a client $t_r$ at the location of facility $r$ with the total demand of clients in $T_r$, i.e., $d_{t_r}=\sum_{j\in T_r}{d_{j}}$. Let $\bar{D}=\{t_1,\cdots,t_l\}$ be the set of our new clients. Now we get a new instance of CKFL, denoted by $I_2$, with facilities $F$ and clients $\bar{D}$.

\textbf{Step 3}. Find an optimal vertex $({x},{y})$ of the feasible region of
the LP-relaxation to the constructed instance $I_2$ in step 2 with $\sum_{i\in F} y_i=l$.

\textbf{Step 4}. We simply open all the facilities with $y_i>0$ in our original instance $I_0$ and then solve a transportation problem to get an integer solution $(x^*,y^*)$.\mbox{}\\

Since we do not know how many facilities are opened in an optimal solution in advance,
we repeat the above 4 steps for $l:=1,\cdots,k$.
Then, output the solution with smallest total cost.

\end{description}
\end{algorithm}

\begin{theorem}\label{theorem:CKFNU}
By Algorithm \ref{CKFL-approx}, each $\alpha$-approximation algorithm for UKM can be extended to get a $(1+2\alpha)$-approximation algorithm for CKFL with uniform opening costs using at most $2k$ facilities.
\end{theorem}
\begin{proof}
Without loss of generality, suppose exactly $k$ facilities are opened in an optimal solution to our original problem (as we check all the cases in our algorithm).

Let $OPT(\ast)$ denote the optimal cost of the instance $\ast$. We consider the following instances.
\begin{center}
  \begin{tabular}{ c | l  }
    $I_0$ &  the original instance.   \\ \hline
    $I_1$ & \tabincell{l}{the constructed instance in Step 1, \\that is a (uncapacitated) $k$-median problem.}  \\ \hline
    $I_2$ & \tabincell{l}{the constructed instance in Step 2 in which \\we have at most $k$ clients.}    \\
  \end{tabular}
\end{center}

Let $COST(\cdot,\cdot)$ be the total cost of obtained solution $(\cdot,\cdot)$. We consider the following solutions
\begin{center}
  \begin{tabular}{ c | l  }
    $(x',y')$ & the obtained integral solution by $\alpha$-approx. alg. for instance $I_1$.   \\ \hline
    $(x,y)$ & \tabincell{l}{an optimal fractional solution of instance $I_2$.}  \\ \hline
    $(x^*,y^*)$ & \tabincell{l}{an integral solution of instance $I_0$ while using at most $2k$ facilities.}    \\
  \end{tabular}
\end{center}

Clearly, we have $COST(x,y)\leq OPT(I_2),$ and $COST(x',y')\leq \alpha OPT(I_1).$

By the process to construct instance $I_2$, we have $OPT(I_0)+COST(x',y')\geq OPT(I_2)$.
Moreover, we know that $OPT(I_0)\geq OPT(I_1)+k f$.

We will prove
\[
COST(x^*,y^*) \leq  COST(x',y') + COST(x,y)+k f.
\]

We first show that we can obtain an integer solution for $I_2$ with the total cost at most $COST(x,y)+k f$ in Step 4 of Algorithm \ref{CKFL-approx}.
We have $|\{i\mid 0<y_i<1\}|\leq k+1$, since Lemma \ref{CKPfractional} still holds when $\sum_{i\in F} y_i=k$.
Moreover, if $|\{i\mid 0<y_i<1\}|>0,$ then $|\{i\mid y_i=1\}|\leq k-1$. So, we open at most $2k$ facilities.
Thus, the total cost of the obtained solution for $I_2$ is at most $COST(x,y)+kf$.

Then, based on the above solution for $I_2$ we can construct an integer solution for $I_0$ by moving the demand of $t_r$, which is located at the same position with facility $r$,
back to all clients in $T_r=\{j\in D^\prime | x^\prime_{r j}=1\}$ with increasing at most $COST(x',y')$ cost as $COST(x',y')=\sum_{r=1}^k {\sum_{j\in D^\prime}{d_j c_{r,j} x^\prime_{r,j}}}$.
Therefore, the solution obtained by Step 4 has $COST(x^*,y^*) \leq  COST(x',y') + COST(x,y)+k f$.

Then, we have
\begin{eqnarray*}
COST(x^*,y^*) &\leq & COST(x',y')+ OPT(I_2)+kf \\
&= & (OPT(I_2)- COST(x',y'))+ 2 COST(x',y')+kf \\
&\leq & OPT(I_0)+ 2 COST(x',y')+kf \\
&\leq & OPT(I_0)+ 2 \alpha OPT(I_1)+kf \\
&\leq & OPT(I_0)+ 2 \alpha (OPT(I_1)+ kf)\leq (1+2\alpha) OPT(I_0).
\end{eqnarray*}
That is, the approximation ratio is $1+2\alpha$.
\end{proof}

We can obtain the following result as there is a $(3+\epsilon)$-approximation algorithm for the (uncapacitated) $k$-median problem in \cite{AryaGKMMP},
and we can make sure that at most $2k-1$ facilities are opened in step 4 if all capacities are equal.
\begin{corollary}
Algorithm \ref{CKFL-approx} can get an integer solution within $7+\epsilon$ times of the optimal cost by using at most $2k$ facilities ($2k-1$ facilities) for the hard capacitated $k$-facility location problem with uniform opening costs (with uniform opening costs and uniform capacities).
\end{corollary}

\subsection{Extension}
We show how to combine the algorithm in Section \ref{sec:CKFLNU_uniform_opening_costs} with the algorithm
of Charikar and Guha \cite{CharikarG} to improve the approximation ratio for the
capacitated facility location problem (CFL) with uniform opening cost.
In this section, we only consider uniform opening cost.
To simplify the description, sometimes we omit ``with uniform opening cost'' when we refer to the problems.

A $(\beta,\delta)$-approximation algorithm for the (uncapacitated) $k$-median problem (UKM)
outputs an integer solution by using at most $\delta k$ facilities, with service cost at most $\beta$ times the optimal total cost.

\begin{theorem}\label{theorem:CFL}
Each $(\beta,\delta)$-approximation algorithm for the $k$-median problem gives rise to a
$\max\{2\beta+1, \delta+1\}$-approximation algorithm for the CFL with uniform opening costs.
\end{theorem}

\begin{proof}
A crucial observation is that
if exactly $k$ facilities are opened in the optimal solution for an instance $I$ of CFL, then
this solution is also an optimal solution to the corresponding instance of CKFL,
where the input is the same as that in $I$ but with an extra constraint that at most $k$ facilities can be opened.
Thus, if for each $k=1,\cdots,n, n=|F|$ we can obtain the optimal solution for CKFL,
then the solution with smallest total cost is the optimal solution for CFL.

Our algorithm for CFL is as follows:
Repeat the 4 steps in Algorithm \ref{CKFL-approx} for $l:=1,\cdots,n$, $n=|F|$. In each iteration, we consider the constraint $\sum_{i\in F} y_i\leq l.$ Then, output the solution with smallest total cost.

To get a better approximation ratio for CFL, we replace the $\alpha$-approximation algorithm for the $k$-median problem in Step 1 of Algorithm \ref{CKFL-approx} by a $(\beta,\delta)$-approximation algorithm.
Then, for each iteration $l$, we obtain an integer solution for the instance $I_1$ with
at most $\delta l$ open facilities.
Thus, we have at most $\delta l$ clients in instance $I_2$.

Again, without loss of generality, suppose exactly $k$ facilities are opened in an optimal solution for the original instance $I_0$.

We maintain all the notations in the proof of Theorem \ref{theorem:CKFNU}.
Notice that we have $COST(x,y)\leq OPT(I_2),$
$OPT(I_0)+COST(x',y')\geq OPT(I_2)$ and $OPT(I_0)\geq OPT(I_1)+k f$ still hold.
Moreover, $COST(x',y')\leq \beta OPT(I_1)$,

We will prove
\[
COST(x^*,y^*) \leq  COST(x',y') + COST(x,y)+\delta k f.
\]

The proof is similar to that in proof of Theorem \ref{theorem:CKFNU}.
We first show that we can obtain an integer solution for $I_2$ with the total cost at most $COST(x,y)+\delta k f$ in step 4.
Note that $|\{i\mid 0<y_i<1\}|\leq \delta k+1$, since Lemma \ref{CKPfractional} still holds when $\sum_{i\in F} y_i=k$.
So, at most $(\delta+1) k$ facilities are opened at the end,
since if $|\{i\mid 0<y_i<1\}|>0,$ then $|\{i\mid y_i=1\}|\leq k-1$.
Thus, the total cost of the obtained solution for $I_2$ is at most $COST(x,y)+ \delta k f$.

Then, by moving the demand of $t_r$
back to all clients in $T_r=\{j\in D^\prime | x^\prime_{r j}=1\}$,
we can construct an integer solution for $I_0$. This operation increases at most $COST(x',y')$ cost as
$COST(x',y')=\sum_{r=1}^k {\sum_{j\in D^\prime}{d_j c_{r,j} x^\prime_{r,j}}}$.
Therefore, the solution obtained by Step 4 has $COST(x^*,y^*) \leq  COST(x',y') + COST(x,y)+ \delta k f$.

Then, we have
\begin{eqnarray*}
COST(x^*,y^*) &\leq & COST(x',y')+  OPT(I_2)+\delta kf\\
&\leq & COST(x',y')+ OPT(I_0)+ COST(x',y')+ \delta kf \\
&= & OPT(I_0)+ 2 COST(x',y')+\delta kf\\
&\leq &  OPT(I_0)+ 2 \beta OPT(I_1)+\delta kf \\
&\leq & OPT(I_0)+ 2 \beta (OPT(I_0)-kf)+\delta kf \\
&\leq & (1+2 \beta) (OPT(I_0)-kf)+(\delta+1) kf.
\end{eqnarray*}
Recall that we assume that exactly $k$ facilities are opened in the optimal solution to $I_0$.
So, the total service cost of the optimal solution to $I_0$ is $OPT(I_0)-kf$.
Then, $COST(x^*,y^*)\leq  \max\{2\beta+1, \delta+1\} OPT(I_0).$
\end{proof}

\begin{theorem}(\cite{CharikarG})\label{theorem:bifactor_UFL}
Let $SOL$ be any solution to the uncapacitated facility location problem (possibly fractional), with facility cost $F_{SOL}$
and service cost $C_{SOL}$. For any $\gamma>0$, the local search heuristic proposed (together with scaling)
gives a solution with facility cost at most $(1+\frac{2}{\gamma})F_{SOL}$ and service cost at most $(1+\gamma)C_{SOL}.$
The approximation is up to multiplicative factors of $(1+\epsilon)$ for arbitrarily small $\epsilon>0.$
\end{theorem}

Based on Theorem \ref{theorem:bifactor_UFL}, we can obtain the following corollary.
\begin{corollary}\label{corollary:UKM}
For any $\gamma>0$, there exists a $((1+\epsilon)(1+\gamma),(1+\epsilon)(1+\frac{2}{\gamma}))$-approximation algorithm for the $k$-median problem, where $\epsilon>0$ can be arbitrarily small.
\end{corollary}
\begin{proof}
Let $(x',y')$ be an optimal solution with total cost $T$ to the LP relaxation of UKM.
A crucial observation is that $(x',y')$
is also a feasible fractional solution to UFL with uniform opening cost $f>0$.
Let $SOL=(x',y')$ with total facility cost $F_{SOL}$
and total service cost $C_{SOL}$.
Note that $C_{SOL}=T$ and $F_{SOL}\leq kf.$

Now, it is easy to see that, based on the Charikar and Guha algorithm \cite{CharikarG}, we could get an
integer solution with at most $(1+\gamma)(1+\epsilon)$ times the optimal cost while using at most  $(1+\frac{2}{\gamma})(1+\epsilon)k$
facilities for the $k$-median problem. That is, there exists a  $((1+\epsilon)(1+\gamma),(1+\epsilon)(1+\frac{2}{\gamma}))$-approximation algorithm for the $k$-median problem, where $\epsilon>0$ can be arbitrarily small.

\end{proof}

The following theorem can be obtained by combining Corollary \ref{corollary:UKM} with Theorem \ref{theorem:CFL} and setting $\gamma=0.78078.$
\begin{theorem}
There is a $4.562(1+\epsilon)$-approximation algorithm for the
capacitated facility location problem with uniform opening costs, where $\epsilon>0$ can be arbitrarily small.
\end{theorem}

\bibliographystyle{alpha}
\bibliography{KFL}

\end{document}